\documentclass[12pt, a4paper]{article}

\usepackage{a4wide}
\usepackage[T2A]{fontenc}
\usepackage[utf8]{inputenc}
\usepackage{amssymb, amsthm}
\usepackage{amsmath}
\usepackage{comment}
\usepackage{color}
\usepackage[pdftex]{graphicx}
\usepackage{hyperref}
\usepackage{mathtools}

\DeclareMathOperator{\sgn}{sgn}
\DeclareMathOperator{\comm}{comm}

\begin{document}

\def\newstr{\par\noindent}
\def\ms{\medskip\newstr}

\renewcommand{\proofname}{Proof}
\renewcommand{\labelenumi}{$\bullet$}

\newcommand{\ifdraft}[1]{#1}
\definecolor{aocolour}{rgb}{0.7,0.8,1}
\definecolor{vmcolour}{rgb}{1,0.8,0.7}
\newcommand{\ao}[1]{\ifdraft{\noindent\colorbox{aocolour}{A.O.: #1}}}
\newcommand{\vm}[1]{\ifdraft{\noindent\colorbox{vmcolour}{V.M.: #1}}}

\newcommand{\Z}{\mathbb{Z}}
\newcommand{\N}{\mathbb{N}}
\newcommand{\R}{\mathbb{R}}
\newcommand{\Q}{\mathbb{Q}}
\newcommand{\K}{\mathbb{K}}
\newcommand{\Cm}{\mathbb{C}}
\newcommand{\Pm}{\mathbb{P}}
\newcommand{\Zero}{\mathbb{O}}
\newcommand{\F}{\mathbb{F}_2}
\newcommand{\ilim}{\int\limits}
\newcommand{\impl}{\Rightarrow}
\newcommand{\set}[2]{\{ \, #1 \mid #2 \, \}}

\newcommand{\A}{\mathcal{A}}
\newcommand{\B}{\mathcal{B}}
\renewcommand{\C}{\mathcal{C}}

\theoremstyle{plain}
\newtheorem{thm}{Theorem}
\newtheorem{lm}{Lemma}
\newtheorem*{st}{Statement}
\newtheorem*{prop}{Property}
\newtheorem{prob}{Problem}
\newtheorem{idea}{Idea}
\newtheorem{conjecture}{Conjecture}

\newtheorem{oldtheorem}{Theorem}
\renewcommand{\theoldtheorem}{\Alph{oldtheorem}}
\newtheorem{oldconj}{Conjecture}
\renewcommand{\theoldconj}{\Alph{oldconj}}

\theoremstyle{definition}
\newtheorem{defn}{Definition}
\newtheorem{ex}{Example}
\newtheorem{cor}{Corollary}

\theoremstyle{remark}
\newtheorem*{rem}{Remark}
\newtheorem*{note}{Note}
\title{Why the equivalence problem for unambiguous grammars 
has not been solved back in 1966?}

\author{Vladislav Makarov\thanks{Saint-Petersburg State University. Supported by Russian Science Foundation, project 18-11-00100.}}

\maketitle

\begin{abstract}
In 1966, Semenov, by using a technique based on power series, 
suggested an algorithm that tells apart the languages described by an unambiguous grammar and a DFA. At the first glance, it may appear that the algorithm can be easily modified to yield a full solution of the equivalence problem for unambiguous grammars. This article shows why this hunch is, in fact, incorrect.
\end{abstract}

\section{Preface}

This Section contains some details about the technical structure of this paper
and the reasons for its existence. Therefore, feel free to skip it. Just remember that this
paper has a ``sequel'': ``Cocke--Younger--Kasami--Schwartz--Zippel algorithm
and its relatives''.

This paper and the companion paper ``Cocke--Younger--Kasami--Schwartz--Zippel algorithm
and its relatives'' are based on the Chapters 3 and 4 of my Master's thesis~\cite{makarov-master} respectively.
While the thesis is published openly in the \href{https://dspace.spbu.ru/handle/11701/30103?locale=en}{SPbSU system}, it has not been published
in a peer-reviewed journal (or via any other scholarly accepted publication method) yet.

These two papers are designed to amend the issue. As of the current date, I have not
submitted them to a refereed venue yet. Therefore, they are only published as arXiv
preprints for now. 

Considering the above, it should not be surprising that huge parts of the original text are copied almost
verbatim. However, the text is not totally the same as the Chapter 3 of my thesis. Some things
are altered for better clarity of exposition and there are even some completely new parts.

Why did I decide to split the results into two papers? There are two main reasons. 

Firstly, both papers are complete works by themselves. From the idea standpoint, some
of the methods and results of the companion paper were motivated by the careful
observation of results of this one. However, the main result of the companion paper
is stated and proven without any explicit references to the content of this paper.

The second reason is closely connected to the particular structure of this paper. 
Specifically, most of the results here are ``negative'' in the following sense.
There is a specific enticing approach to the equivalence problem for unambiguous grammars.
This approach \emph{does not work}, not without any serious modifications at least. However,
despite that, I have never seen any discussion about said approach in the previously published research
on the topic. This paper aims to fill this gap in the literature.

I realize that it \emph{may}
happen that all of the following is something that is already known to the specialists in the
field. However, I sincerely believe that this paper deserves to be made easily accessible
to the wider mathematical community. There are many known cases in the mathematics
when something is a ``folklore'' result that is never properly published and is only circulated
via informal communication between researchers. 

There are two most common cases
when such a thing can happen. The first one is a known error in some important
work in some research field. The second one is when there is an informal knowledge
that some approach to a problem does not work (usually, it is hard-to-impossible
to formalize such statements without losing most of their ``power''). 
In both these cases,
a young and inexperienced researcher that wants to start working in the field can waste
a lot of time: by searching for an error in the first case and by futilely searching for a way
to apply the known method in the problem. 

What makes the situation even more
difficult is that in both these cases it may be difficult to publish the result in a reasonably
prestigious peer-reviewed venue due to the lack of notability. Indeed, either way, the
result you are trying to publish is not new and can therefore be rejected. Which would
be fine if you could actually find any paper that contain the result you wanted to publish!
This creates a paradoxical situation when something is well-known, but is only circulated
via personal communication. 

As you may have guessed, I find the aforementioned situation saddenning. Hence, I
think that one still should publish all such ``negative'' results somehow. 
While a peer-reviewed venue would be ideal, a simple arXiv preprint that will never
become a ``proper'' paper due to being well-known is also OK. Hence, my second
reason for split publication can be explained in the following way.  By splitting the 
``negative'' and 	the ``positive'' results, I create a single paper that can be
properly published and a single paper with a less clear publication status, which
can still be shared with the international mathematical community via arXiv
even if it will not get a ``proper'' publication.

With technical details out of the way, let us move on to the mathematical part.

\section{The equivalence problem for unambiguous grammars}

Consider the following problem: 

\begin{prob}[The equivalence problem for unambiguous grammars.] \label{grail}
You are given two ordinary grammars $G_1$ and $G_2$. Moreover, you know  that
they are both unambiguous from a $100\%$ trustworthy source. Is there 
an algorithm to tell whether $L(G_1)$ and $L(G_2)$ are equal? Because the grammars
are guaranteed to be unambiguous, the algorithm may behave arbitrarily if either of $G_1$
and $G_2$ is ambiguous, including not terminating at all.
\end{prob} 
\begin{rem} The wording is so specific for a reason: it is undecidable to tell whether
given ordinary grammar is unambiguous or not. Hence, it is impossible to somehow plug in the verifier of unambiguity into the algorithm. In complexity-theoretic terms, this is a \emph{promise} problem.
However, these technicalities are not so important now, because we are nowhere near to the solution for them to matter.
\end{rem}

The same problem for arbitrary \emph{ordinary} grammars is undecidable, but all the known
proofs use extremely ambiguous grammars. 

Around the turn of the millenium, there was a major breakthrough on Problem~\ref{grail}: Senizergues proved that the equivalence
problem for deterministic grammars is decidable~\cite{dpda}.  
Still, most unambiguous grammars \emph{are not} deterministic. Moreover, Senizergues's proof is extremely long and difficult: the original
paper is 159 pages long! The proof relies on complicated arguments about deterministic
pushdown automata (the definitions of deterministic grammars and deterministic pushdown automata are out of the scope for this paper). Essentially, it shows that some very carefully constructed first-order theory is
complete. There were some simplifications since then~\cite{simpler-dpda}, but the proof still remains
complicated.

Personally, I believe that the answer to Problem~\ref{grail} should be ``Yes''. Hence,
I will look at the problem from a more algorithmic side. More or less all results of this and
the following papers have been inspired by this problem somehow, but their \emph{statements}
sometimes are not directly related to the Problem~\ref{grail}.

A naive approach to the Problem~\ref{grail} would be to enumerate all words of length
at most $f(G_1, G_2)$, where $f$ is some computable function, and, for each of them, 
check whether
it belongs to $L(G_1)$ and $L(G_2)$ by standard cubic-time parsing algorithm. However,
how should we choose the function $f$, exactly? For arbitrary 
ordinary grammars the corresponding problem is undecidable and, therefore,
the required function $f$ is uncomputable. On the other hand, an existence of such $f$ 
for unambiguous grammars is equivalent to the fact that Problem~\ref{grail}
is decidable. Indeed, if the problem is decidable, then we can check $G_1$ and $G_2$
for equivalence. If they are equivalent, then we can define $f(G_1, G_2)$ arbitrarily. 
Otherwise, just iterate over all words until you find the first difference between $L(G_1)$
and $L(G_2)$. So, looking at the problem in this way is not helpful at all.
\begin{defn} A word $w$ is a difference between languages $L_1$ and $L_2$ if $w$ is in 
one of them, but not in the other. The \emph{first difference} is the lexicographically smallest
of the shortest differences.
\end{defn}

The above approach works badly when the first difference between $L(G_1)$ and $L(G_2)$
has large length. Maybe, there is an approach that may work well even in the hypothetical
case when the first difference has \emph{very} large length?
And there is! Of course, it suffers from its own issues, and I am not even remotely close to solving Problem~\ref{grail},
but I think that the methods and the results I managed to obtain are interesting enough.

Without loss of generality, we can assume that $G_1$ and $G_2$ are in Chomsky normal
form and do not contain the empty word. For a grammar $G$ in Chomsky normal form,
define its \emph{size} $|G|$ simply as the number of rules in the grammar.

Equivalence of unambiguous grammars is directly related to the question of emptiness of a given
GF(2)-grammar. If the emptiness of a GF(2)-grammar is decidable, then so is the equivalence of unambiguous grammars. 

\section{Semenov's approach}

The methods of this paper are directly inspired by the way Semenov~\cite{semenov} approached a simpler case of Problem~\ref{grail}. 
Moreover, it may appear on a first glance, that a simple variation on the Semenov's idea actually solves the Problem~\ref{grail}, but the devil is in the details.

For the sake of completeness, I will explain the whole approach of Semenov here. It will not
take much space and will prove crucial later.
All the exposition in this Section is based on the Semenov's paper~\cite{semenov}. 

Let us prove the following theorem (in order to showcase all the necessary details the proof is not the
simplest one):
\begin{oldtheorem}[\cite{semenov}]\label{subset} Given two unambiguous grammars $G_1$ and $G_2$, such that $L(G_1) \subset L(G_2)$, one can algorithmically check whether $L(G_1) = L(G_2)$.
\end{oldtheorem}

The difference between Problem~\ref{grail} and Theorem~\ref{subset} is the very strong
condition that $L(G_1)$ is a subset of $L(G_2)$. To prove Theorem~\ref{subset}, let us 
follow the following simple plan:
\begin{enumerate}
\item[1.] Replace languages with formal power series by ignoring the order of letters.
Of course, this is not an equivalent transformation, but we will not lose too much information.
\item[2.] Translate grammars $L(G_1)$ and $L(G_2)$ into a system of polynomial equations 
over some power series.
\item[3.] Power series are equal if and only if they are equal in all points from an arbitrarily small neighbourhood of origin. Hence, we can replace polynomial equations for power series with
quantified polynomial equations for real numbers. By Tarski--Seidenberg theorem, the last
problem is decidable.
\end{enumerate}
 
To proceed with the first step of the plan, consider \emph{commutative images} of $L(G_1)$ and 
$L(G_2)$.
Informally, a commutative image of a language corresponds to interpreting a language as a sum
of its words and then forgetting that the letters do not actually commute. Formally, the definition
is the following:
\begin{defn} Let $L$ be a language over the alphabet $\Sigma = \{a_1, a_2, \ldots, a_k\}$. By its \emph{commutative image}
$\comm(L)$, I mean the formal power series over variables $a_1$, $a_2$, \ldots, $a_k$ 
(yes, the variables are the letters of $\Sigma$), with the
coefficient before $a_1^{d_1} a_2^{d_2} \ldots a_k^{d_k}$ being the number of words with exactly $d_i$ letters $a_i$ for each $1 \leqslant i \leqslant k$.
\end{defn}
Consider two examples:
\begin{ex}
$\comm(\{ab, ba, b^2 a, c^2 de, bab\}) = ab + ba + b^2 a + c^2 de + bab = 2ab + 2ab^2 + c^2 de$ 
\end{ex}
\begin{ex}
Let $L$ be the language of correct bracket sequences, but with letters $a$ and $b$
instead of symbols ``(`' and ``)'' for clarity. Then, 
$\comm(L) = \sum\limits_{n=0}^{+\infty} C_n a^n b^n$, where $C_n = \dfrac{(2n)!}{n! \cdot (n +1)!}$ are Catalan's numbers. Indeed, in this example, all words of length $2n$ from
the language contain $n$ letters $a$ and $n$ letters $b$.
\end{ex}

Very importantly, $f(K \sqcup L) = f(K) + f(L)$ and $f(K \cdot L) = f(K) f(L)$, as long as concatenation
$K \cdot L$ is unambiguous.

It turns out that comparing commutative images is actually \emph{much} easier than 
comparing the languages themselves:
\begin{oldtheorem}[\cite{semenov}]\label{comm_images} Given two unambiguous grammars $G_1$ and 
$G_2$, there is an algorithm for checking whether $\comm(L(G_1))$ and $\comm(L(G_2))$ are equal.
\end{oldtheorem}
\begin{proof}[Proof of Theorem~\ref{comm_images}] Without loss of generality, both $G_1 = (\Sigma, N_1, R_1, S_1)$
and $G_2 = (\Sigma, N_2, R_2, S_2)$ are in Chomsky's normal form and do not contain the empty word.

Consider one of the grammars, for example $G_1 = ( \Sigma, N_1, R_1, S_1)$.
For each nonterminal $C$ of $G_1$, define $f(C) \coloneqq \comm(L(C))$. 
Because $G_1$
is an unambiguous grammar, all concatenations are unambiguous and all unions are disjoint.
Therefore,
\begin{equation}\label{series_semenov}
f(C) = \sum\limits_{(C \to DE) \in R_1} (f(D) \cdot f(E)) + \sum\limits_{(C \to a) \in R_1} a
\end{equation}
Here, the sums range over all rules that correspond to the nonterminal $C$: the first sum is over
all ``normal'' rules $C \to DE$ and the second sum is over all ``final'' rules $C \to a$ with
$a \in \Sigma$ (recall that each element of $\Sigma$ can be interpreted as a variable).

This can be interpreted as a system of polynomial equations over ``indeterminates'' $f(C)$.
The \emph{real} values of $f(C)$ (that is, $\comm(L(C))$) satisfy those equations. 

Write down those systems for $G_1$ and $G_2$. We need to check whether $\comm(L(G_1)) = \comm(L(G_2))$ or, in other words, $f(S_1) = f(S_2)$. A grammar in Chomsky's normal form with $p$ rules has at most $p^{2n}$ parse trees for strings of length $n$~\cite[Lemma 1]{semenov}.
Hence, $f(S_1)$ and $f(S_2)$ both converge as power series of several variables in the interior of a ball with radius $1/(\max(|R_1|, |R_2|)^{2})$ and the center in the origin. In particular, they
both converge when all $|\Sigma|$ variables take values that are less than $\varepsilon \coloneqq 1/(\max(|R_1|, |R_2|)^2 \cdot |\Sigma|)$.

It is known that two formal power series are equal as series as long as they are equal as functions in all points of a small neighbourhood of origin. To check all of them at the same time, we can use	universal quantifiers over reals. To be precise, the following statements are equivalent:
\begin{itemize}
\item[1.] The commutative images $\comm(L(G_1))$ and $\comm(L(G_2))$ are equal.
\item[2.] For all ways to assign real values to the elements of $\Sigma \sqcup N_1 \sqcup N_2$
(the alphabet letters and the nonterminals of $G_1$ and $G_2$),
$(\mathrm{Small} \wedge \mathrm{Correctness}) \Rightarrow (S_1 = S_2)$. Here, by 
$\mathrm{Small}$ I mean the finite conjunction of conditions $|a| < \varepsilon$ for $a \in \Sigma$. 
By $\mathrm{Correctness}$, I mean that all internal grammar equations like Equation~\ref{series_semenov} are satisfied. Check the following Example~\ref{stupid} for better understanding.
\end{itemize}
Finally, telling whether the second statement is true or not is a special case of Tarski--Seidenberg
theorem about decidability of first-order theory of reals.
\end{proof}
\begin{ex}\label{stupid} Consider the two following simple unambiguous grammars over the alphabet $\Sigma = \{a, b\}$ that generate languages $\{a, ab\}$ and $\{a, ba\}$ with equal commutative images:
\begin{align*}
S_1 &\to A_1 B_1 \sqcup a			          &S_2 &\to B_2 A_2 \sqcup a \\
A_1 &\to a                                                &A_2 &\to a \\
B_1 &\to b                                                &B_2 &\to b 
\end{align*}
The corresponding quantified statement is $\forall a, b, S_1, A_1, B_1, S_2, A_2, B_2 \in \R \colon ((|a| < \varepsilon) \wedge (|b| < \varepsilon) \wedge (S_1 = A_1 B_1 + a) \wedge (A_1 = a) \wedge (B_1 = b) \wedge (S_2 = A_2 B_2 + a) \wedge (A_2 = a) \wedge (B_2 = b)) \Rightarrow (S_1 = S_2)$. Here, $\varepsilon = 1/(|\Sigma| \cdot \max(|R_1|, |R_2|)^2) = 1/(2 \cdot 4^2) = 1/32$ and $|a| < \varepsilon$ can be rewritten as $(a \cdot 32 < 1) \wedge (-a \cdot 32 < 1)$. So, in the end, this
is a first-order sentence over reals with only universal quantifiers. This special case is much easier than the general case of Tarski--Seidenberg theorem from the computational perspective~\cite{exists_real}.
\end{ex}

The remaining part of the proof is simple:

\begin{proof}[Proof of Theorem~\ref{subset}] 
If $L(G_1) = L(G_2)$, then $\comm(L(G_1)) = \comm(L(G_2))$.
Otherwise, $L(G_2)$ strictly contains $L(G_1)$ and at least one coefficient of $\comm(L(G_2))$
is strictly greater than the corresponding coefficient of $\comm(L(G_1))$. Hence, $L(G_1) = L(G_2)$ if and only if $\comm(L(G_1)) = \comm(L(G_2))$.
\end{proof}

\section{Matrix substitution and polynomial identities}\label{section_matrix}

Clearly, the argument from previous Section did not use commutativity 
of real number multiplication that much. What we used instead are some other
properties of real numbers: that the equality of power series over $\R$ follows from pointwise equality and that the first-order theory of real numbers is decidable.

So, we want to replace real numbers with something that captures noncommutativity
of string concatenation at least to some extent, but without losing the decidability property.
Matrices with real entries seem like a good middle ground: they do not commute, but
their addition and multiplication is defined by polynomial equations over their entries.

Indeed, if $A$, $B$ and $C$ are real $d \times d$ matrices, then $A = BC$, by definition, means
that $A_{i, j} = \sum\limits_{k=1}^d B_{i, k} C_{k, j}$ for all $i$ and $j$ from $1$ to $d$.
So, if $d$ is fixed, the condition $A = BC$ can be expressed as a conjunction of $d^2$ 
polynomial equations over \emph{real numbers}. Similarly, the condition 
$A = B + C$ is also a big conjunction in disguise. Finally, a good matrix equivalent
of $|A| < \varepsilon$ is ``$\ell^1$ norm of $A$ is less than $\varepsilon$'', or, in other
words, $\sum\limits_{i=1}^d \sum\limits_{j=1}^d |A_{i, j}| < \varepsilon$.

Hence, we can apply the same line of reasoning that we did before. Fix some number $d$, possibly depending on $G_1$ and $G_2$ (but in a computable way). 
Write down a first-order formula
akin to one from Example~\ref{stupid}, but with matrices instead of real numbers.  Then, split 
every equation into basic equations like $A = BC$ and $A = B + C$ by introducing extra variables.
Finally, replace each matrix variable with $d^2$ real variables corresponding to its entries
and replace equations like $A = BC$ and $A = B + C$ with big conjunctions, as seen above. 
The
result is still some universal first-order statement about real numbers. We can check whether
it is true or not.
This way, we have noncommutativity of matrices at our disposal, without sacrificing decidability.

Clearly, this approach can lead only to false positives (languages are different, but we could not tell them apart), but not to false negatives. A false positive for languages $L_1$ and $L_2$ corresponds
to the fact that $L_1 \neq L_2$, but $\sum\limits_{w \in L_1} X_{w_1} X_{w_2} \ldots X_{w_{|w|}} = \sum\limits_{w \in L_2}  X_{w_1} X_{w_2} \ldots X_{w_{|w|}}$ for any way to choose $|\Sigma|$ real matrices with small norm ~--- one matrix $X_a$ for each letter $a \in \Sigma$. After cancelling
out common words, we are left with nontrivial (that is, not $0 = 0$) equation
$\sum\limits_{w \in L_1 \setminus L_2} X_{w_1} X_{w_2} \ldots X_{w_{|w|}} - \sum\limits_{w \in L_2 \setminus L_1}  X_{w_1} X_{w_2} \ldots X_{w_{|w|}} = 0$. 
Finally, a known homogeneity-based argument~\cite[Chapter 4]{pi_rings} allows to ``split'' this single equation by degree to get a separate equality for each word length. Precisely, for all $n \geqslant 0$,
\begin{equation}\label{rep}
\sum\limits_{w \in (L_1 \setminus L_2) \cap \Sigma^n} X_{w_1} X_{w_2} \ldots X_{w_{n}} - \sum\limits_{w \in (L_2 \setminus L_1) \cap \Sigma^n}  X_{w_1} X_{w_2} \ldots X_{w_{n}} = 0
\end{equation}
\begin{defn} For a language $L$, its \emph{$n$-slice} is the language $\set{w}{w \in L, |w| = n}$ of all words from $L$ of length exactly $n$.
\end{defn}
If $L_1 \neq L_2$, then, for some $n$, their $n$-slices are different as well. Then, the corresponding Equation~\eqref{rep} of degree $n$ is nontrivial and, by homogeneity, is true for all real matrices
and not only those of small norm.

On the first glance, it appears that this is a solution of Problem~\ref{grail}.
Indeed, it seems intuitive that there is no \emph{single} nontrivial matrix equation that is true for \emph{all} $d \times d$ matrices for $d \geqslant 2$. However, this intuition is dead wrong.

\begin{oldtheorem}[Amitsur--Levitsky theorem~\cite{pi_rings}] For any $d \times d$ matrices $X_1$, $X_2$, \ldots, $X_{2d}$ over any commutative ring, 
\begin{equation}\label{amitsur_identity} 
\sum\limits_{\sigma \in S_{2d}} (-1)^{\sgn(\sigma)} X_{\sigma(1)} X_{\sigma(2)} \ldots X_{\sigma(2d)} = 0
\end{equation}
\end{oldtheorem}

Things like the left-hand side of the Equation~\eqref{amitsur_identity} are called
\emph{polynomial identities}. Formally,

\begin{defn} A polynomial $p$ in $n$ \emph{noncommuting} variables
is a \emph{polynomial identity} for $d \times d$ matrices if and only if $p(A_1, A_2, \ldots, A_n) = 0$
for any $d \times d$ matrices $A_1$, $A_2$, \ldots, $A_n$. 
\end{defn}
\begin{note} Polynomial identities are slightly misleadingly named, because usually the word
``polynomial'' refers to polynomials in commuting variables. However, this is standard terminology.
\end{note}

Moreover, polynomial identities are common enough~\cite[Chapter 3]{pi_rings} to make ruling them out one-by-one impossible. Of course, we are only interested  in polynomial identities
where the coefficients before each monomial is in the set $\{-1, 0, +1\}$ (only thise can arise from comparing unambiguous grammars), but there still is quite a lot of those. 

\section{What \emph{can} we do with matrix substitution?}

Of course, not everything is so bleak. Firstly, it is reasonable to expect that polynomial identities for large $d$ are pretty complicated and will not appear accidentally. 
This means that matrix substitution with 
small constant $d$ is a very good heuristic for Problem~\ref{grail}. For even better results, handle
all small lengths of possible differences with the main theorem of the companion paper.

It is not just a heuristic, though. The simplest possible measure of ``complicatedness'' is the number of monomials. And, indeed, it is known that all polynomial identities for $d \times d$
matrices must contain at least $2^{d-1}$ noncommutative monomials~\cite[Theorem 2]{many_monomials}. 
\begin{defn} Languages $L_1$ and $L_2$ are \emph{$d$-similar} if matrix substitution with $d \times d$
matrices cannot tell them apart. In particular, equal languages are $d$-similar for all $d$.
\end{defn}
The above formalization of ``complicatedness'' immediately leads to the following result:
\begin{thm}\label{bounded_diff} Let $G_1$ and $G_2$ be unambiguous grammars over the alphabet $\Sigma$, satisfying conditions $L(G_1) \neq L(G_2)$ and $|(L(G_1) \triangle L(G_2)) \cap \Sigma^n| < 2^{d-1}$ for all $n \geqslant 0$. Then, $L(G_1)$ and $L(G_2)$ \emph{are not} $d$-similar.
\end{thm}
\begin{proof} Because $L(G_1) \neq L(G_2)$,  their $n$-slices are different for some $n$. Then, the corresponding 
polynomial identity~\ref{rep} has $|(L(G_1) \setminus L(G_2)) \cap \Sigma^n| + |(L(G_2) \setminus L(G_1)) \cap \Sigma^n| = |(L(G_1) \triangle L(G_2)) \cap \Sigma^n| < 2^{d-1}$ monomials and cannot be an identity for $d \times d$ matrices.
\end{proof}
\begin{rem} In fact, we just proved a slightly stronger, but more awkward statement: $L(G_1)$ and $L(G_2)$ cannot be $d$-similar if there \emph{exists} such $n$, that $n$-slices of $L(G_1)$ and $L(G_2)$ differ, but in less than $2^{d-1}$ strings.
\end{rem}

The statement of Theorem~\ref{bounded_diff} is interesting in the following way: normally, one would
expect that the case of close languages to be the hardest one for Problem~\ref{grail}. However, this is not the case: as Theorem~\ref{bounded_diff} shows, the languages can be \emph{too close} to be $d$-similar! Hence, this immediately leads to a solution of Problem~\ref{grail} for close languages

\textbf{Warning!} The rest of this paper is highly speculative in a sense that it details a possible way to solve Problem~\ref{grail}, but there are big obstacles for pretty much every step of the plan. If it is not your cup of tea, it makes sense to skip directly to the next CSection. However, if you choose to do so, you probably may want to skip to the conclusion (Section~\ref{chap_concl}).

Of course, there is no reason for $d$ to be a constant. It may depend on $G_1$ and $G_2$, but in a computable way. This leads to a, admittedly, extremely incomplete plan of attack on Problem~\ref{grail}. We will need the notion of T-ideal. Informally, T-ideals are closer under taking consequences. Formally,
\begin{defn} A set $I$ of nonncommutative polynomials is a T-ideal, if and only if
\begin{itemize} 
\item $0 \in T$,
\item for any $p, q \in T$, their sum $p + q$ is also in $T$,
\item for any $p \in T$ and any (not necessarily from $T$) noncommutative polynomial $q$, both their product $pq$ and $qp$ are in $T$,
\item for any noncommutative polynomial $p \in T$ in $n$ variables, and any (not necessarily from $T$) noncommutative polynomials $q_1$, $q_2$, \ldots, $q_n$, the result $p(q_1, q_2, \ldots, q_n)$ of substituting $q_i$ in place of variables is also in $T$. For example, if $X_1 X_2 - X_2 X_1$ is in $T$,
then $(AB + A)(A + BA) - (A + BA)(AB + A) = (ABA + ABBA + AA + ABA) - (AAB + AA + BAAB + BAA) = (ABBA - BAAB) + (2ABA - BAA - AAB)$ is also in $T$.
\end{itemize}
The T-ideal $I$ is generated by the set $X$, if $I$ is the smallest T-ideal that contains $X$ as a subset.
\end{defn}

For simplicity, let us assume the following well-known conjecture about the structure of polynomial
identities for matrices:
\begin{oldconj}[Razmyslov's conjecture~\cite{conjecture_structure}] \label{like_det} All polynomial identities for $d \times d$ matrices lie in a T-ideal generated by
Amitsur-Levitsky identity (Equation~\eqref{amitsur_identity}) and the following identity:
\begin{equation}\label{other_basis}
\sum\limits_{\sigma \in S_d} (-1)^{\sgn(\sigma)} [X_1^{\sigma(1)}, X_2] [X_1^{\sigma(2)}, X_2] \ldots [X_1^{\sigma(d)}, X_2] = 0,
\end{equation}
where $[A, B]$ denotes the commutator of $A$ and $B$: $[A, B] \coloneqq AB - BA$.
\end{oldconj}

\begin{defn}\label{noncomm_det} For a $n \times n$ matrix $X$ with not necessarily commuting entries its \emph{noncommutative determinant} is defined as $\sum\limits_{\sigma \in S_n} X_{1, \sigma(1)} X_{2, \sigma(2)} \ldots X_{n, \sigma(n)}$.
 \end{defn}

Visually, both Equation~\eqref{amitsur_identity} and~\eqref{other_basis} resemble the definition
of determinant. The Definition~\ref{noncomm_det} suggests that it is possible to give a useful
interpretation to this similarity. Denote the left-hand side of Equation~\eqref{amitsur_identity} by $h_1 = h_1 (X_1, \ldots, X_{2d})$ and the left-hand side of Equation~\eqref{other_basis} by $h_2 = h_2 (X_1, X_2)$.
Hence, we can use the following idea (I do not care about time complexity here,
because there are too many obstacles even when decidability only is concerned):
\begin{idea}\label{kill_grammars} Suppose that $L(G_1) \neq L(G_2)$. Then, let $\ell$ be the length of the first difference between $L(G_1)$ and $L(G_2)$. Then, let $d_{\textrm{max}}$ be the maximal such $d$, that $\ell$-slices of $L(G_1)$ and $L(G_2)$ are  $d$-similar. By Conjecture~\ref{like_det}, 
the identity for $\ell$-slices can be written as $s_1 h_1 (p_{1, 1} \cdot \ldots p_{1, 2d_{\textrm{max}}}) \cdot r_1 + s_2 \cdot h_1 (p_{2, 1} \ldots p_{2, 2d_{\textrm{max}}}) \cdot r_2 + \ldots +  s_k \cdot h_1 (p_{k, 1}, \ldots, p_{k, 2d_{\textrm{max}}}) \cdot r_k + \ell_{k+1} \cdot h_2 (q_{1, 1}, q_{1, 2}) \cdot r_{k+1} + \ldots + \ell_{k+m} \cdot h_2 (q_{m,1}, q_{m, 2}) \cdot r_{k + m}$, where $s_i$, $r_i$, $p_{i, j}$ are some noncommutative polynomials and $k$ and $m$ are some nonnegative integers.
There are three possible cases:
\begin{enumerate}
\item[1.] $d_{\textrm{max}}$ is small, say, $d_{\textrm{max}} < 2^{2^{|G_1| + |G_2|}}$. In this case, matrix substitution with $d \coloneqq 2^{2^{|G_1| + |G_2|}}$ works. Hence, proving that
the other two cases cannot actually happen solves Problem~\ref{grail}.
\item[2.] $d_{\textrm{max}}$ is large, but not when compared to $\ell$. Say, $10d_{\textrm{max}}^{10} < \ell$ and $d_{\textrm{max}} \coloneqq 2^{2^{|G_1| + |G_2|}}$. In this case, we can try to apply pumping lemma or a similar style argument. Specifically, each monomial from the our polynomial identity has large length, but splits up into \emph{huge} chunks that correspond to a monomial in one of $p_{i, j}$, $q_{i, j}$, $s_i$ or $r_i$. Hence, it is possible to pump the internals of those big chunks pretty much separately. Therefore, we get a lot of possibilities to ``disrupt'' the polynomial identities in $m$-slices with $m > \ell$. Unfortunately, I do not know any good way of implementing this idea.
\item[3.] $d_{\textrm{max}}$ is large, and is comparable to $\ell$.  Say, $10d_{\textrm{max}}^{10} \geqslant \ell$ and $d_{\textrm{max}} \coloneqq 2^{2^{|G_1| + |G_2|}}$. In this case, recall that the identity is a sum of $k + m$ detrerminant-like things. If $k + m = 1$ and all polynomials
$p_{i, j}$, $q_{i, j}$ satisfy some technical requirements, it is possible to extract a small arithmetic circuit for noncommutative determinant out of the grammars $G_1$ and $G_2$ (and small arithmetic circuits for noncommutative determinant are extremely unlikely to exist, because noncommutative permanent is $\mathrm{\#P}$-complete~\cite[Theorem 3.5]{determ}). I believe that it should be
possible to extend the technique to the case of small $k + m$,
 but, again, I do not know what
to do when $k + m$ is large, for example, $k + m > 2^{d_{\textrm{max}} / 10}$.
\end{enumerate}
\end{idea}

In the end, there are some major obstacles to both steps of the plan (proving the impossibility of situations 2 and 3 in the above). Ideally, we need a way to restrict our consideration only to
``simple enough'' polynomial identities (both Equation~\eqref{amitsur_identity} and Equation~\eqref{other_basis} are simple enough by themselves, but their consequences are not). Then, everything would work out in the end.

\section{Conclusion}\label{chap_concl}

Back in 1966, Semenov solved an important special case of Problem~\ref{grail} ---
the case when one of the grammars is a regular grammar. From the ideological 
standpoint, his approach is pretty simple. Hence, it may appear at the first glance
that  a simple modification of the method will lead to the solution of the full problem.

However, this is \emph{very much} not the case because of the relative prevalence
of matrix polynomial identities. Despite that, we still can get some partial results like Theorem~\ref{bounded_diff} more-or-less directly from the matrix substitution method. 
Moreover, I propose a potential way to ``attack'' Problem~\ref{grail}. Admittedly,
the plan appears to be not very realistic: it \emph{both} relies on a several unproven
conjectures \emph{and} has a lot of extremely unclear steps that require
developing completely new methods in order to be completed.

This is all for the ``negative'' results for now. The ``sequel'' paper presents some
``positive'' results that are inspired by the same ideas, but are different enough
to the point you can understand the majority of the companion paper without
reading this one.


\begin{thebibliography}{00}

\bibitem{many_monomials}
	V. Arvind, P. Mukhopadhyay, S. Raja,
	\href{https://doi.org/10.1145/3055399.3055442}
	{``Randomized polynomial time identity testing for noncommutative circuits''},
	\emph{STOC 2017: Proceedings of the 49th annual ACM symposium on Theory of computing},
	831--841.

\bibitem{pi_rings}
      V. S. Drensky,	
	 \emph{Free algebras and PI-algebras: graduate course in algebra},
	 Springer-Verlag, 1996.

\bibitem{conjecture_structure} V. S. Drensky,
	{``A minimal basis for identities of a second-order matrix algebra 
	over a field of characteristic 0''},
	\emph{Algebra Logika},
	20:3 (1981), 282--290.

\bibitem{determ} 
	S. Chien, P. Harsha, A. Sinclair, S. Srinivasan
	\href{https://doi.org/10.1145/1993636.1993703}
	{``Almost Settling the Hardness of Noncommutative Determinant''}
	\emph{STOC 2011: Proceedings of the 43rd annual ACM symposium on Theory of computing},
	499--508.
	
\bibitem{simpler-dpda} P. Jancar,
	{``Decidability of DPDA Language Equivalence via First-Order Grammars''},
	\emph{Proceedings -- Symposium on Logic in Computer Science},
	(2012)

\bibitem{exists_real} H. Joos, R. Marie-Francoise, S. Pablo,
	\href{https://doi.org/10.1093/comjnl/36.5.427}
	{`'On the theoretical and practical complexity of the existential theory of reals"},
	\emph{The Computer Journal},
	36:5 (1993), 427--431.
         
\bibitem{makarov-master} V. Makarov,
	\href{https://dspace.spbu.ru/bitstream/11701/30103/1/Makarov_23_07.pdf}
	{``Algebraic and analytic methods for grammar ambiguity'''},
	\emph{Research Repository Saint Petersburg State University},
	June 2021
	
\bibitem{semenov} A. L. Semenov, 
	{``Algorthmic Problems for Power Series and Context-free Grammars''}, 
	\emph{Soviet Mathematics}, 
	14 (1973), 1319

\bibitem{dpda} G. S\'enizergues, 
	\href{http://dx.doi.org/10.1016/S0304-3975(00)00285-1}
	{``$L(A)=L(B)$? decidability results from complete formal systems''},
	\emph{Theoretical Computer Science},
	251:1--2 (2001), 1--166.

\end{thebibliography}
\end{document}